\documentclass[10pt,journal,twocolumn]{IEEEtran}
\usepackage[cmex10]{amsmath}
\usepackage{amssymb}
\usepackage{enumerate}
\usepackage{mathrsfs}
\usepackage{epsfig}
\usepackage{yhmath}
\usepackage{mathtools}
\def\mb{\mathbb}
\def\m{\mathcal}
\def\md{\mathbb}
\def\ms{\mathscr}
\def\mb{\mathbf}

\def\tn{\textnormal}
\def\wt{\widetilde}
\def\wh{\widehat}

\def\NaturalF{\md{N}}
\newcommand{\dfn}{ \stackrel{\tn{def}}{=} }

\newtheorem{lemma}{Lemma}
\newtheorem{corollary}{Corollary}
\newtheorem{theorem}{Theorem}

\newtheorem{remark}{Remark}

\newtheorem{example}{Example}

\allowdisplaybreaks[3]

\begin{document}

\title{A Note on a Characterization of R\'enyi Measures and its Relation to Composite Hypothesis Testing}

\author{Ofer~Shayevitz
\thanks{The author is with the Information Theory \& Applications Center, University of California, San Diego, USA \{email: ofersha@ucsd.edu\}. }}

\maketitle

\thispagestyle{empty}

\begin{abstract}
The R\'enyi information measures are characterized in terms of their Shannon counterparts, and properties of the former are recovered from first principle via the associated properties of the latter. Motivated by this characterization, a two-sensor composite hypothesis testing problem is presented, and the optimal worst case miss-detection exponent is obtained in terms of a R\'enyi divergence.
\end{abstract}

\thispagestyle{empty}

\section{Introduction}
The Shannon Entropy and the Kullback-Leibler divergence play a pivotal role in the study of information theory, large deviations and statistics, arising as the answer to many of the fundamental questions in these fields. Besides their operational importance, these quantities also possess some very natural properties one would expect an information measure to satisfy, a fact that has spurred several different axiomatic characterizations, see \cite{Csiszar2008} and references therein.

Motivated by the axiomatic approach, R\'enyi suggested a more general class of measures satisfying some slightly weaker postulates, yet still intuitively appealing as measures of information \cite{Renyi1960}. Remarkably, this ``reversed'' line of thought has proved fruitful; the R\'enyi information measures have been shown to admit several operational interpretations, thereby ``justifying'' their definition. Among other cases, the R\'enyi entropy has appeared as a fundamental quantity in problems of source coding with exponential weights \cite{Campbell65}, random search \cite{Renyi65}, error exponents in source coding \cite{jelinek68}, generalized cutoff rates for source coding \cite{Csiszar95}, guessing moments \cite{arikan96}, privacy amplification \cite{Bennet_etal1995}, predictive channel coding with transmitter side information \cite{ErezZamir2001}, and redundancy-delay exponents in source coding \cite{Delay_constrained_coding_IT}. The R\'enyi divergence has emerged (sometimes implicitly) in the analysis of channel coding error exponents \cite{Gallager65,PolyanskiyVerdu2010}, generalized cutoff rates for hypothesis testing \cite{Csiszar95}, multiple source adaptation \cite{Mansour09}, and generalized guessing moments \cite{Erven10}. Several different definitions of a R\'enyi mutual information (and the associated capacity) were tied to generalized cutoff rates in channel coding \cite{Arimoto77,Csiszar95}, and to distortion in joint source-channel coding \cite{Ingber_etal2008}.

Interestingly, even though the Shannon measures are a special case of the R\'enyi measures, the latter can admit a variational characterization in terms of the former. For the R\'enyi entropy (of order $\alpha<1$) this has been observed in the context of guessing moments \cite{arikan96,MerhavArikan99}, and for one definition of a R\'enyi mutual information, has been derived in the context of generalized cutoff rates in channel coding \cite[Appendix]{Csiszar95}. In this note, relations of that type and their applications\footnote{In fact, the impetus for this short study grew out of a recent work by the author and colleagues \cite{Delay_constrained_coding_IT}, where the characterization for the R\'enyi entropy of order $2$ has been utilized to obtain a lower bound on the redundancy-delay exponent in lossless source coding.} are further examined. Section \ref{sec:prelim} contains the necessary mathematical background. In Section \ref{sec:char}, a variational characterization for the various R\'enyi measures via the Shannon measures is provided. In Section \ref{sec:prop}, it is demonstrated how properties of the R\'enyi measures can be derived in a very instructive (and sometimes simpler) fashion directly from their variational characterization, via the associated properties of their Shannon counterparts. Finally, the discussed characterization motivates the study of a two-sensor composite hypothesis testing problem in which the R\'enyi divergence is shown to play a fundamental role, yielding a new operational interpretation to that quantity. This observation is discussed in Section \ref{sec:hyp}.

\section{Preliminaries}\label{sec:prelim}
\subsection{Shannon Information Measures}
Let $\m{X}$ be a finite alphabet, and denote by $\ms{P}(\m{X})$ the set of all probability distributions over $\m{X}$. The support of a distribution $P\in\ms{P}(\m{X})$ is the set $S(P)\dfn \{x\in\m{X} : P(x)>0\}$. The (Shannon) \textit{entropy} of $P\in\ms{P}(\m{X})$ is\footnote{We use the conventions $0\log 0 = 0$, and $a\log\frac{a}{0} = 0 \;\text{or}\; +\infty$ according to whether  $a=0$ or $a>0$ respectively.}
\begin{equation*}
H(P) \dfn -\sum_{x\in\m{X}} P(x)\log P(x)\,.
\end{equation*}
The (Kullback-Leibler) \textit{divergence} between two distributions $P_1,P_2\in\ms{P}(\m{X})$ is
\begin{equation*}
D(P_1\|P_2) \dfn \sum_{x\in\m{X}} P_1(x)\log\left(\frac{P_1(x)}{P_2(x)}\right)\,.
\end{equation*}
We write $P_1\ll P_2$ to indicate that $S(P_1)\subseteq S(P_2)$. Note that $D(P_1\|P_2)<\infty$ if and only if $P_1\ll P_2$.

Let $\m{X},\m{Y}$ be two finite alphabets. A \textit{channel} $\m{W}:\m{X}\mapsto\m{Y}$ is a set of probability distributions $\{W(\cdot|x)\in\ms{P}(\m{Y})\}_{x\in\m{X}}$ that maps a distribution $P\in\ms{P}(\m{X})$ to the distributions $P\circ W\in\ms{P}(\m{X}\times\m{Y})$ and $PW\in\ms{P}(\m{Y})$, according to
\begin{align*}
(P\circ W)(x,y) &\dfn P(x)W(y|x) \\
PW(y) &\dfn \sum_{x\in\m{X}}P(x)W(y|x)\,.
\end{align*}
For any two channels $V:\m{X}\mapsto\m{Y}, W:\m{X}\mapsto\m{Y}$, we write
\begin{equation*}
D(V\|W|P) \dfn \sum_{x\in\m{X}}P(x)D(V(\cdot|x)\|W(\cdot|x))
\end{equation*}
The \textit{(Shannon) mutual information} associated with $P$ and $W$ is
\begin{align}
\nonumber I(P,W) &\dfn H(PW) - \sum_{x\in\m{X}}P(x)H(W(\cdot|x))\\
& = \min_Q\sum_{x\in\m{X}}P(x)D(W(\cdot|x)\|Q)\label{eq:MI4I}\\
&  =\min_QD(P\circ W\|P\times Q)\label{eq:MI4K}
\end{align}
where the identities are well known. The (Shannon) \textit{capacity} of a channel $W$ is
\begin{align*}
C(W) \dfn \max_P I(P,W)
\end{align*}

A distribution $P\in\ms{P}(\m{X})$ induces a product distribution $P^n\in\ms{P}(\m{X}^n)$, where $P^n(x^n)\dfn \prod_{k=1}^nP(x_k)$. The \textit{type} of a sequence $x^n\in\m{X}^n$ is the distribution $\pi_{x^n}\in\ms{P}(\m{X})$ corresponding to the relative frequency of symbols in $x^n$. The set of all possible types of sequences $x^n$ is denoted $\ms{P}^n(\m{X})$. The \textit{type class} of any type $Q\in\ms{P}^n(\m{X})$ is the set $T_Q\dfn \{x^n\in\m{X}^n: \pi_{x^n}=Q\}$.

The following facts are well known \cite{csiszar_korner}.
\begin{lemma}\label{lem:types}
For any type $Q\in\ms{P}^n(\m{X})$ and any $x^n\in T_Q$:
\begin{enumerate}[(i)]
\item $P^n(x^n) = 2^{-n(D(Q\|P)+H(P))}$. \label{item:prob_in_type}
\item $|\ms{P}^n(\m{X})|^{-1}2^{nH(Q)} \leq |T_Q| \leq  2^{nH(Q)}$. \label{item:type_size}
\item $|\ms{P}^n(\m{X})|  = {n+|\m{X}|-1 \choose |\m{X}|-1} \leq (n+1)^{|\m{X}|}$. \label{item:num_of_types}
\item For any $\delta>0$ \label{item:large_dev}
\begin{equation*}
P^n\left(\left\{x^n\in\m{X}^n : D(\pi_{x^n}\|P) \geq  \delta\right\}\right) \leq |\ms{P}^n(\m{X})|2^{-n\delta}\,.
\end{equation*}
\end{enumerate}
\end{lemma}

\subsection{R\'enyi Information Measures}
Let $\alpha>0$, $\alpha\neq 1$ throughout. The \textit{R\'enyi entropy of order $\alpha$} of a distribution $P\in\ms{P}(\m{X})$ is
\begin{equation*}
H_\alpha(P) \dfn \frac{1}{1-\alpha} \log\sum_{x\in\m{X}}P(x)^\alpha\,.
\end{equation*}
We denote by $H_0(P),H_1(P)$ and $H_\infty(P)$ the limits of $H_\alpha(P)$ as $\alpha$ tends to $0,1$ and $\infty$, respectively\footnote{These limits are known to exist, a fact we reestablish in the sequel.}. \newcounter{temp}\setcounter{temp}{\value{footnote}}
The \textit{R\'enyi divergence of order $\alpha$} between two distributions $P_1,P_2\in\ms{P}(\m{X})$ is\footnote{For $\alpha>1$ we adopt the convention where $a^\alpha \cdot 0^{1-\alpha} = 0  \;\text{or}\; +\infty$ according to whether  $a=0$ or $a>0$ respectively.}
\begin{equation*}
D_\alpha(P_1\|P_2) \dfn \frac{1}{\alpha-1} \log\sum_{x\in \m{X}}P_1(x)^\alpha P_2(x)^{1-\alpha}\,.
\end{equation*}
We denote by $D_0(P_1\|P_2)$, $D_1(P_1\|P_2)$ and $D_\infty(P_1\|P_2)$ the limits of $D_\alpha(P_1\|P_2)$ as $\alpha$ tends to $0,1$ and $\infty$, respectively\footnotemark[\value{temp}]. Note that for $\alpha<1$, $D_\alpha(P_1\|P_2) < \infty$ if and only if $S(P_1)\cap S(P_2)\neq \emptyset$, and for $\alpha>1$, $D_\alpha(P_1\|P_2) < \infty$ if and only if $P_1\ll P_2$.

The R\'enyi equivalent of the Shannon mutual information has several different definitions, each generalizing a different expansion of the latter, see \cite{Csiszar95} and references therein. Here we discuss the following two alternatives:
\begin{align}\label{eq:I_a_def}
I_\alpha(P,W) \dfn \min_Q\sum_{x\in\m{X}}P(x)D_\alpha(W(\cdot|x)\|Q)
\end{align}
corresponding to (\ref{eq:MI4I}), and
\begin{equation}\label{eq:K_a_def}
K_\alpha(P,W)\dfn \min_Q D_\alpha(P\circ W\|P\times Q)
\end{equation}
corresponding to (\ref{eq:MI4K}). Following \cite{Csiszar95}, we define the \textit{capacity of order $\alpha$} of $W$ via (\ref{eq:I_a_def}), i.e.,
\begin{align*}
C_\alpha(W) \dfn \max_P I_\alpha(P,W)
\end{align*}
However, using (\ref{eq:K_a_def}) in the definition yields the same capacity function \cite{Csiszar95}, a fact we reaffirm in the sequel.

\section{Characterization}\label{sec:char}
In this section, we derive the basic characterization for the various R\'enyi measures in terms of the Shannon measures.
\begin{theorem}\label{thrm}
For $\alpha>1$,
\begin{align}
H_\alpha(P) &= \min_{Q}\left\{\frac{\alpha}{\alpha-1}\,D(Q\|P)+H(Q)\right\}\label{eq:H_a}\\
D_\alpha(P_1\|P_2) &= \max_{Q\ll P_1}\left\{\frac{\alpha}{1-\alpha}\,D(Q\|P_1)+D(Q\|P_2)\right\}\label{eq:D_a}\\
I_\alpha(P,W) &= \max_V\left\{I(P,V) + \frac{\alpha}{1-\alpha}D(V\|W|P)\right\}\label{eq:I_a}
\\
K_\alpha(P,W) &= \max_{Q}\left\{I_\alpha(Q,W)+\frac{1}{1-\alpha}D(Q\|P)\right\}\label{eq:K_a}
\end{align}
For $\alpha<1$, replace $\min$ with $\max$ and vice versa.
\end{theorem}
\begin{remark}
The $\alpha<1$ counterpart of (\ref{eq:H_a}) is mentioned in \cite{arikan96,MerhavArikan99}. Both (\ref{eq:H_a}) and (\ref{eq:D_a}) are simple generalizations, for which we provide an elementary proof. Relation (\ref{eq:I_a}) can be found in \cite[Appendix]{Csiszar95}, however here we provide a slightly different proof directly via (\ref{eq:D_a}). Relation (\ref{eq:K_a}) appears to be new.
\end{remark}

\begin{proof}
Let $\m{X}_1\dfn S(P_1)$ and $\m{X}_2\dfn S(P_2)$ for short. We derive a characterization for the functional
\begin{equation}
J_{\alpha,\beta}(P_1,P_2)\dfn -\log\sum_{x\in\m{X}_1}P_1(x)^\alpha P_2(x)^\beta
\end{equation}
for any $\alpha>0$ and $\beta$. This will yield (\ref{eq:H_a}) and (\ref{eq:D_a}) in particular, and will also prove useful in the sequel. It is readily verified that the functional is additive, i.e., $J_{\alpha,\beta}(P_1^n,P_2^n) = nJ_{\alpha,\beta}(P_1,P_2)$. Therefore,
\begin{align*}
&J_{\alpha,\beta}(P_1,P_2) = -\frac{1}{n}\log\sum_{x^n\in\m{X}_1^n}P_1(x^n)^\alpha P_2(x^n)^\beta
\\
&\leq  -\frac{1}{n}\log\sum_{Q\in\ms{P}^n(\m{X}_1)}2^{-n(\alpha (D(Q\|P_1)+H(Q)) + \beta(D(Q\|P_2)+H(Q))} \\
&\qquad\qquad\qquad\times |\ms{P}^n(\m{X}_1)|^{-1}2^{nH(Q)}
\\
&\leq \!\min_{Q\in\ms{P}^n(\m{X}_1)}\!\left\{\alpha D(Q\|P_1)+\beta D(Q\|P_2) + (\alpha+\beta-1)H(Q)\right\}
\\
&\qquad\qquad\qquad+\frac{|\m{X}_1|\log{(n+1)}}{n}
\end{align*}
where properties (\ref{item:prob_in_type}) and (\ref{item:type_size}) of Lemma \ref{lem:types} were used in the first inequality, and property (\ref{item:num_of_types}) was used in the second inequality. Similarly,
\begin{align*}
&J_{\alpha,\beta}(P_1,P_2)
\\
&\geq  -\frac{1}{n}\log\sum_{Q\in\ms{P}^n(\m{X}_1)}2^{-n(\alpha D(Q\|P_1)+\beta D(Q\|P_2) + (\alpha+\beta-1)H(Q))}
\\
&\geq \!\min_{Q\in\ms{P}^n(\m{X}_1)}\!\left\{\alpha D(Q\|P_1)+\beta D(Q\|P_2) + (\alpha+\beta-1)H(Q)\right\}
\\
&\qquad\qquad\qquad -\frac{|\m{X}_1|\log{(n+1)}}{n}\,.
\end{align*}
$\bigcup_n\ms{P}^n(\m{X}_1)$ is dense in $\ms{P}(\m{X}_1)$, and the objective function is continuous in $Q$ over the compact set $\ms{P}(\m{X}_1\cap\m{X}_2)$, and equals $\pm\infty$ over $\ms{P}(\m{X}_1)\setminus \ms{P}(\m{X}_1\cap\m{X}_2)$ according to ${\rm sign}(\beta)$. Thus, taking the limit as $n\rightarrow\infty$, we obtain:
\begin{align}\label{eq:J}
&J_{\alpha,\beta}(P_1,P_2)
\\
\nonumber &\;= \min_{Q\ll P_1}\left\{\alpha D(Q\|P_1)+\beta D(Q\|P_2) + (\alpha+\beta-1)H(Q)\right\}\,.
\end{align}
The statement for $H_\alpha(P)$ (resp. $D_\alpha(P_1\|P_2)$) now follows by substituting $\beta=0$ (resp. $\beta=1-\alpha$), normalizing by $\alpha-1$ (resp. $1-\alpha$), and noting the possible change in sign that replaces $\min$ with $\max$. For $H_\alpha(P)$, taking the $\min$ or $\max$ over all $Q\in\ms{P}(\m{X})$ does not change anything.

We now turn to prove (\ref{eq:I_a}) and (\ref{eq:K_a}). As in \cite{Csiszar95}, the minimum in (\ref{eq:I_a_def}) and (\ref{eq:K_a_def}) can be replaced with an infimum over distributions $Q$ with $S(Q)=\m{Y}$, merely excluding possibly infinite values. This will be implicit below. For $\alpha>1$, we have
    {\allowdisplaybreaks
    \begin{align}\label{eq:I_der}
    \begin{multlined}
    I_\alpha(P,W)\\
    \shoveleft{\stackrel{(a)}{=}\inf_Q\sum_{x\in\m{X}}P(x)\!\max_{R\ll W(\cdot|x)}\!\left(\frac{\alpha}{1-\alpha}D(R\|W(\cdot|x)) + D(R\|Q)\!\right)}\\
    \shoveleft{= \inf_Q\max_{V}\sum_{x\in\m{X}}P(x)\left(\frac{\alpha}{1-\alpha}D(V(\cdot|x)\|W(\cdot|x))\right.}\\
    \shoveright{\left.\vphantom{\frac{\alpha}{1-\alpha}}+ D(V(\cdot|x)\|Q)\right)}\\
    \shoveleft{\stackrel{(b)}{=} \max_{V}\inf_Q\left(\frac{\alpha}{1-\alpha}D(V\|W|P)\right.}\\
     \shoveright{\left.\vphantom{\frac{\alpha}{1-\alpha}}+ \sum_{x\in\m{X}}P(x)D(V(\cdot|x)\|Q)\right)}\\
     \shoveleft{\stackrel{(c)}{=}\max_V\left\{I(P,V) + \frac{\alpha}{1-\alpha}D(V\|W|P)\right\}}\\ \vspace{-10pt}
    \end{multlined}\vspace{-15pt}
    \end{align}}
The maximization is taken over all channels $V$ such that $P\circ V \ll P\circ W$. The equalities above are justified as follows:
\begin{enumerate}[(a)]
\item by virtue of Theorem \ref{thrm}.
\item the objective function is continuous and concave\footnote{Concavity in $V$ follows by writing each of the summands as $\left[D(V(\cdot|x)\|Q)-D(V(\cdot|x)\|W(\cdot|x))\right]+\frac{1}{1-\alpha}D(V(\cdot|x)\|W(\cdot|x))$, which is the sum of a linear function and a concave function in $V$ (for $\alpha>1$). } in $V$ over a compact set for any fixed $Q$, and convex in $Q$ for any fixed $V$. Hence, $\max$ and $\inf$ can be interchanged \cite[Theorem 4.2]{Sion1958}.
\item on account of (\ref{eq:MI4I}).
\end{enumerate}
This establishes (\ref{eq:I_a}) for $\alpha>1$.\footnote{Taking the last $\max$ over all channels $V:\m{X}\mapsto\m{Y}$ changes nothing.} The simpler derivation for $\alpha<1$ is similar.

To establish (\ref{eq:K_a}), write:
{\allowdisplaybreaks
\begin{align}\label{eq:K_der}
\begin{multlined}
K_\alpha(P,W) \\
\shoveleft[20pt]{\stackrel{(a)}{=} \inf_Q\max_{P'\circ V}\left\{\frac{\alpha}{1-\alpha}D(P'\circ V\|P\circ W)\right.} \\
\shoveright{\left.\vphantom{\frac{\alpha}{1-\alpha}}+D(P'\circ V\|P\times Q)\right\}} \\
\shoveleft[20pt]{\stackrel{(b)}{=} \max_{P'\circ V}\inf_Q\left\{\frac{\alpha}{1-\alpha}D(P'\circ V\|P\circ W)\right.} \\
\shoveright{\left.\vphantom{\frac{\alpha}{1-\alpha}}+D(P'\circ V\|P\times Q)\right\}} \\
\shoveleft[20pt]{= \max_{P'\circ V}\inf_Q\left\{\frac{\alpha}{1-\alpha}D(P'\circ V\|P\circ W)+D(P'\|P)\right.}\\
\shoveright{\left.\vphantom{\frac{\alpha}{1-\alpha}} + D(P'\circ V\|P'\times Q)\right\}}\\
\shoveleft[20pt]{\stackrel{(c)}{=} \max_{P'\circ V}\left\{\frac{\alpha}{1-\alpha}D(P'\circ V\|P\circ W)+D(P'\|P)\right.}
\\
\shoveright{\left.\vphantom{\frac{\alpha}{1-\alpha}} + I(P',V)\right\}}\\
\shoveleft[20pt]{= \max_{P'\circ V}\left\{\frac{\alpha}{1-\alpha}D(V\|W|P')+\frac{1}{1-\alpha}D(P'\|P)\right.}
\\
\shoveright{\left.\vphantom{\frac{\alpha}{1-\alpha}}+ I(P',V)\right\}}\\
\shoveleft[20pt]{\stackrel{(d)}{=} \max_{P'}\left\{I_\alpha(P',W)+\frac{1}{1-\alpha}D(P'\|P)\right\}}\\ \vspace{-15pt}
\end{multlined}\vspace{-10pt}
\end{align}}
The maximization is over all $P'$ and $V$ such that $P'\circ V \ll P\circ W$. Equalities (a) and (b) are justified similarly to their counterparts in (\ref{eq:I_der}), while (c) and (d) follows from (\ref{eq:MI4K}) and (\ref{eq:I_a}) respectively. This establishes (\ref{eq:K_a}) for $\alpha>1$.\footnote{Taking the last $\max$ over all $P'\in\ms{P}(\m{X})$ changes nothing.} The simpler derivation for $\alpha<1$ is similar.
\end{proof}

\section{Properties Revisited}\label{sec:prop}
In this section, we derive some well known and lesser known properties of the R\'enyi measures directly via the characterization in Theorem \ref{thrm}, and the associated properties of the Shannon measures. These alternative derivations appear in many cases more instructive than a direct proof, and are sometimes simpler.

\subsection{$H_\alpha(P)$}\label{subsec:H}

For convenience, define:
\begin{align*}
G_\alpha(P; Q) &\dfn \frac{\alpha}{\alpha-1}D(Q\|P) + H(Q)\,.
\end{align*}
We will repeatedly use the fact that by Theorem \ref{thrm}, $G_\alpha(P; Q)$ is an upper (resp. lower) bound for $H_\alpha(P)$ for $\alpha>1$ (resp. $\alpha<1$). Without loss of generality, we will restrict $Q\ll P$ in Theorem \ref{thrm} throughout.

\begin{enumerate}[1.]
\item \textit{$H_\alpha(P)$ is a non-increasing function of $\alpha$.} \label{item:H_dec}

\qquad\textit{Proof:} For any fixed $Q$, $G_\alpha(P; Q)$ is non-increasing in $\alpha$ over $(0,1)$ (resp. $(1,\infty)$). By Theorem \ref{thrm}, $H_\alpha(P)$ is the maximum (resp. minimum) of $G_\alpha(P; Q)$ taken over $Q$, hence it is also non-increasing in $\alpha$ over $(0,1)$ (resp. $(1,\infty)$). To order the two regions, we note that for $\alpha<1$
\begin{equation*}
H_\alpha(P) \geq G_\alpha(P;P) = \frac{\alpha}{\alpha-1}D(P\|P) + H(P) = H(P)
\end{equation*}
and similarly for $\alpha>1$ we have $H_\alpha(P) \leq H(P)$.

\item $H_\alpha(P)$ is concave in $P$ for $\alpha<1$.

\qquad\textit{Proof:} $H(Q)$ is concave in $Q$ and $D(Q\|P)$ is convex in $(P,Q)$, hence $G_\alpha(P;Q)$ is concave in $(P,Q)$ for $\alpha<1$. The statement follows since maximizing a concave function over a convex set ($\ms{P}(S(P))$ in this case) preserves concavity.

\item $H_0(P) = \log{|S(P)|}$. \label{item:H_0}

\qquad\textit{Proof:} Let $Q'$ be the uniform distribution over $S(P)$. Then on the one hand,
\begin{align*}
H_0(P) &\geq \lim_{\alpha\rightarrow 0}\left(\frac{\alpha}{\alpha-1}D(Q'\|P) + H(Q')\right)
\\
&= H(Q') = \log{|S(P)|}
\end{align*}
and on the other hand,
\begin{align*}
H_0(P) &= \lim_{\alpha\rightarrow 0}\max_{Q\ll P}\left\{\frac{\alpha}{\alpha-1}D(Q\|P) + H(Q)\right\}
\\
&\leq \max_{Q\ll P}H(Q) = \log{|S(P)|}\,.
\end{align*}

\item $H_\infty(P)  = -\log{\max_{x\in\m{X}} P(x)}$: \label{prop:H_infty}

\textit{Proof:} Let $Q'(x')=1$, where $x'\in\m{X}$ satisfies $P(x')=\max_{x\in\m{X}}P(x)$. Then on the one hand,
\begin{align*}
H_\infty(P) &\leq \lim_{\alpha\rightarrow\infty}\left\{\frac{\alpha}{\alpha-1}D(Q'\|P) + H(Q')\right\}
\\
& = D(Q'\|P) = -\log P(x') = -\log{\max_{x\in\m{X}} P(x)}
\end{align*}
and on the other hand,
\[
\begin{multlined}
\shoveleft{H_\infty(P) \geq \lim_{\alpha\rightarrow\infty}\left(\min_{Q\ll P}\left\{D(Q\|P) + H(Q)\right\} \right.}
\\
\shoveright{\left.\qquad\qquad+ \min_{Q\ll P}\left\{\frac{D(Q\|P)}{\alpha-1} \right\}\right)}
\\
\shoveleft[\widthof{$H_\infty(P)$}]{= \min_{Q\ll P}\left\{D(Q\|P) + H(Q)\right\}}
\\
\shoveleft[\widthof{$H_\infty(P)$}]{= \min_{Q\ll P}\left(-\sum_{x\in\m{X}}Q(x)\log P(x) \right)}
\\
\shoveleft[\widthof{$H_\infty(P)$}]{= -\log{\max_{x\in\m{X}} P(x)}\,.}\\
\end{multlined}
\]\vspace{-20pt}

\item $H_1(P) = H(P)$\label{prop:H_1}

\qquad\textit{Proof:} We consider the limit $\alpha\rightarrow 1^+$, the other limit follows similarly and coincides. We have already seen that for $\alpha>1$, $H_\alpha(P) \leq G_\alpha(P;P) = H(P)$. Intuitively, $Q=P$ must be set in $G_\alpha$ as above, since otherwise the divergence terms blows up. Precisely, fix some $r>H(P)$ and define $M_\alpha\dfn \{Q : \frac{\alpha}{\alpha-1}D(Q\|P) \leq r\}$. Then
\begin{align*}
H_\alpha(P)  &= \lim_{\alpha\rightarrow 1^+} \inf_{Q\in M_\alpha}\left\{\frac{\alpha}{\alpha-1}D(Q\|P)+H(Q)\right\}
\\
&\geq \lim_{\alpha\rightarrow 1^+} \inf_{Q\in M_\alpha}H(Q) = H(P)\,.
\end{align*}
where the last equality holds since $\sup_{Q\in M_\alpha}D(Q\|P)\rightarrow 0$ as $\alpha\rightarrow 1^+$.

\item The general inequality $H_\alpha(P) \leq \frac{\alpha}{\alpha-1}D(Q\|P) + H(Q)$ for $\alpha>1$ and $Q\ll P$ (and its reversed counterpart for $\alpha<1$) is equivalent to the \textit{log-sum inequality}. Specifically, a uniform $Q$ corresponds to the \textit{arithmetic-geometric mean inequality}.

\qquad\textit{Proof:} By direct computation.

\item Let $\ell:\m{X}\mapsto\NaturalF$ be a \textit{codelength assignment} associated with some uniquely decodable code for $P$. Define the exponentially weighted average codelength with parameter $\lambda>0$ for associated with $(P,\ell)$ to be\footnote{Note that $\lambda\rightarrow 0$ yields the usual average codelength criterion, and $\lambda\rightarrow\infty$ yields the maximal codelength criterion.}
    \begin{equation}\label{eq:min_cl}
    \m{L}_\lambda(P,\ell) \dfn \frac{1}{\lambda}\log\sum_{x\in\m{X}}P(x)2^{\lambda\ell(x)}\,.
    \end{equation}
    Then the optimal codelength satisfies:
    \begin{equation*}
    H_\frac{1}{1+\lambda}(P) \leq \min_{\ell} \m{L}_\lambda(P,\ell) \leq H_\frac{1}{1+\lambda}(P) + 1\,.
    \end{equation*}

\qquad\textit{Proof:} We reestablish this result from \cite{Campbell65} via our approach. Define the probability distribution $R(x) \dfn 2^-{\ell(x)}\slash c$, where $c\dfn \sum_x 2^{-\ell(x)}\leq 1$ by Kraft's inequality. Then
\begin{equation}
\m{L}_\lambda(P,\ell) = -\log{c} + \frac{1}{\lambda}\log\sum_{x\in\m{X}}P(x)R(x)^{-\lambda}\,.
\end{equation}
Let $\m{\wh{L}}_\lambda(P,R)$ be the second summand above. When minimizing over all distributions $R$, it is clearly sufficient to take the infimum over those with $S(R)=S(P)$, which for brevity will be implicit below. Hence:
\begin{align*}
&\min_{R\in\ms{P}(\m{X})}\m{\wh{L}}_\lambda(P,R) = \inf_R\m{\wh{L}}_\lambda(P,R)
\\
&\quad\stackrel{(a)}{=} \inf_R\max_{Q\ll P} \left\{-\lambda^{-1}D(Q\|P) + D(Q\|R) + H(Q)\right\}
\\
&\quad\stackrel{(b)}{=} \max_{Q\ll P} \inf_R\left\{-\lambda^{-1}D(Q\|P) + D(Q\|R) + H(Q)\right\}
\\
&\quad\stackrel{(c)}{=} \max_{Q\ll P} \left\{-\lambda^{-1}D(Q\|P) + H(Q)\right\} = H_{\frac{1}{1+\lambda}}(P)\,.
\end{align*}
The equalities are justified as follows:
\begin{enumerate}[(a)]
\item on account of (\ref{eq:J}), by setting $\alpha=1$ and $\beta=-\lambda$.
\item the objective function is concave\footnote{The first summand is concave in $Q$, while the sum of the last two is linear.} and continuous in $Q$ over the compact set $\ms{P}(S(P))$ for any fixed $R$, and convex in $R$ for any fixed $Q$. Hence, $\max$ and $\inf$ can be interchanged \cite[Theorem 4.2]{Sion1958}.
\item by virtue of Theorem \ref{thrm}.
\end{enumerate}
This immediately establishes the lower bound. The associated saddle point is therefore $(Q^*,Q^*)$, where $Q^*$ is the optimizing distribution for $H_{\frac{1}{1+\lambda}}(P)$, hence $\m{\wh{L}}_\lambda(P,Q^*) = H_{\frac{1}{1+\lambda}}(P)$. Plugging $\ell(x) = \lceil-\log Q^*(x)\rceil $ in (\ref{eq:min_cl}) establishes the upper bound.

\item\label{prop:opt_D} The unique optimizing distribution for $G_\alpha(P; Q)$ is
\begin{equation*}
Q^*(x) = \frac{P(x)^\alpha}{\sum_{x\in\m{X}}P(x)^\alpha}\,.
\end{equation*}

\qquad\textit{Proof:} Verify by substitution that $G_\alpha(P; Q^*) = H_\alpha(P)$. Uniqueness follows from strict convexity (resp. concavity) of $G_\alpha(P; Q)$ in $Q$ over $\ms{P}(S(P))$ for $\alpha>1$ (resp. $\alpha<1$).

\item (\textit{Approximate recursivity})
Suppose $P'$ is obtained from $P$ by combining the symbols $x_1,x_2$ (with probabilities $P(x_1)=p_1$ and $P(x_2) = p_2$) into a single symbol $x_1$, i.e., $P'(x_1) = p_1+p_2$ and $P'(x_2)=0$, while retaining all other probabilities. Then\footnote{For binary distributions $P = (p,1-p)$ and $Q=(q,1-q)$, we write $H_\alpha(p) = H_\alpha(P)$ and $D_\alpha(p\|q) = D_\alpha(P\|Q)$.}
\begin{equation*}
H_\alpha(P)  = H_\alpha(P') + c\cdot H_\alpha\left(\frac{p_1}{p_1+p_2}\right)
\end{equation*}
where $c$ satisfies
\begin{equation}\label{eq:c_lb}
\left(p_1^\alpha+p_2^\alpha\right)\cdot 2^{(\alpha-1)H_\alpha(P)} \leq c\leq (p_1+p_2)^\alpha\cdot 2^{(\alpha-1)H_\alpha(P')}
\end{equation}
for $\alpha>1$, and the reversed inequalities for $\alpha<1$. Note that $0\leq c \leq 1$, and $c\rightarrow p_1+p_2$ as $\alpha\rightarrow 1$.

\qquad\textit{Proof:} We prove for $\alpha>1$, the derivation for $\alpha<1$ is similar with the inequalities reversed. Let $Q^*$ minimize $G_\alpha(P;Q)$, and write $Q^*(x_1)=q_1^*, Q^*(x_1)=q_2^*$. Let $Q'$ be obtained from $Q^*$ by combining $x_1,x_2$ as above. Then:
\[
\begin{multlined}
H_\alpha(P') \leq G_\alpha(P';Q')
\\
\shoveleft{= \frac{\alpha}{\alpha-1}\!\left(\!D(P\|Q^*) \!- \! (q^*_1+q^*_2)D\!\left(\frac{q^*_1}{q^*_1+q^*_2}\|\frac{p_1}{p_1+p_2}\right)\!\right)}
\\
\shoveright{+ H(Q^*) - (q^*_1+q^*_2)H\left(\frac{q^*_1}{q^*_1+q^*_2}\right)}
\\
\shoveleft{\leq  H_\alpha(P) - (q_1^*+q_2^*)H_\alpha\left(\frac{p_1}{p_1+p_2}\right)\,.}\\
\end{multlined}\vspace{-10pt}
\]
The recursivity properties of the Shannon entropy and the Kullback-Leibler divergence were used in the equality transition. The last inequality follows by applying Theorem \ref{thrm} twice, and using the definition of $Q^*$. Appealing to Property \ref{subsec:D}.\ref{prop:opt_D} above, the lower bound in (\ref{eq:c_lb}) is established.

For the upper bound, let $Q'^*$ minimize $G_\alpha(P';Q)$. Let the distribution $Q$ be obtained from $Q'^*$ by splitting the probability $Q'^*(x_1)$ between $x_1$ and $x_2$ such that $\frac{Q(x_1)}{Q(x_1)+Q(x_2)} = \frac{p_1^\alpha}{p_1^\alpha+p_2^\alpha}$, while retaining all other probabilities. The bound follows by expanding the inequality $H_\alpha(P)\leq G_\alpha(P;Q)$ as above, using recursivity, Theorem \ref{thrm} and Property \ref{subsec:D}.\ref{prop:opt_D}.

\end{enumerate}

\subsection{$D_\alpha(P_1\|P_2)$}\label{subsec:D}
For convenience, define:
\begin{align*}
G_\alpha(P_1,P_2; Q) &\dfn \frac{\alpha}{1-\alpha}D(Q\|P_1) + D(Q\|P_2)\,.
\end{align*}
We will repeatedly use the fact that by Theorem \ref{thrm}, $G_\alpha(P_1,P_2; Q)$ is a lower  (resp. upper) bound for $D_\alpha(P_1\|P_2)$, for $\alpha>1$ (resp. $\alpha<1$) and any $Q\ll P_1$.

\begin{enumerate}[1.]
\item \textit{$D_\alpha(P_1\|P_2)$ is an increasing function of $\alpha$.}

\qquad\textit{Proof:} Similar to Property \ref{subsec:H}.\ref{item:H_dec}, by noting that $G_\alpha(P_1,P_2; P_1) = D(P_1\|P_2)$.

\item $D_\alpha(P_1\|P_2) \geq 0$ with equality if and only if $P_1=P_2$.

\qquad\textit{Proof:} For $\alpha<1$ this follows immediately from Theorem \ref{thrm} using the same property of $D(P_1\|P_2)$. For $\alpha>1$ use also the monotonicity property above.

\item $D_\alpha(P_1\|P_2)$ is convex in $P_2$ for $\alpha>1$ and any fixed $P_1$, and is convex in the pair $(P_1,P_2)$ for  $\alpha<1$. \label{prop:D_a_convex}

\qquad\textit{Proof:} $D(Q\|P_2)$ is convex in $P_2$ for any fixed $Q$, hence so is $G_\alpha(P_1,P_2; Q)$. The statement for $\alpha>1$ follows since a pointwise maximum of convex functions is convex. For $\alpha<1$, the convexity of $D(Q\|P_1)$ in $(Q,P_1)$ and of $D(Q\|P_2)$ in $(Q,P_2)$ implies that $G_\alpha(P_1,P_2; Q)$ is convex in $(P_1,P_2,Q)$. The result now follows since minimizing a convex function over a convex set ($\ms{P}(S(P_1))$ in this case) preserves convexity.

\item $D_0(P_1\|P_2) = -\log{P_2(S(P_1))}$.

\textit{Proof:} Let $Q'$ be $P_2$ restricted to $S(P_1)$, with the proper normalization. Then on the one hand,
\begin{align*}
D_0(P_1\|P_2) &\leq  \lim_{\alpha\rightarrow 0}\left(\frac{\alpha}{\alpha-1}D(Q'\|P_1) + D(Q'\|P_2)\right)
\\
&= D(Q'\|P_2) = -\log{P_2(S(P_1))}
\end{align*}
and on the other hand,
\begin{align*}
D_0(P_1\|P_2) &= \lim_{\alpha\rightarrow 0}\min_{Q\ll P_1}\left\{\frac{\alpha}{1-\alpha}D(Q\|P_1) + D(Q\|P_2)\right\}
\\
&\geq \min_{Q\ll P_1}D(Q\|P_2) = D(Q'\|P_2)
\\
&= -\log{P_2(S(P_1))}\,.
\end{align*}

\item $D_\infty(P_1\|P_2) = \log{\max_{x\in S(P_2)} \frac{P_1(x)}{P_2(x)}}$

\qquad\textit{Proof:} Let $Q'(x')=1$, where $x'\in\m{X}$ satisfies $P_1(x')\slash P_2(x')=\max_{x\in S(P_2)}\left(P_1(x)\slash P_2(x)\right)$.  The proof is now similar to that of Property \ref{subsec:H}.\ref{prop:H_infty}.

\item $D_1(P_1\|P_2) = D(P_1\|P_2)$

\quad \textit{Proof:} $Q=P_1$ must be set to avoid a blowup of the first divergence term in $G_\alpha(P_1,P_2; Q)$. The proof is similar to that of Property \ref{subsec:H}.\ref{prop:H_1}.

\item (\textit{Data Processing Inequality}) For any pair of distributions \label{prop:DP_div} $P_1,P_2\in\ms{P}(\m{X})$ and channel $W:\m{X}\mapsto\m{Y}$,
    \begin{equation*}
    D_\alpha(P_1W\|P_2W)\leq D_\alpha(P_1\|P_2)\,.
    \end{equation*}

\qquad\textit{Proof:} We prove only for $\alpha<1$.\footnote{This holds for any $\alpha>0$, however the case of $\alpha>1$ does not seem to follow elegantly from our representation, and can be proved directly.} Let $Q^*$ minimize $G_\alpha(P_1,P_2;Q)$. Write:
\begin{align*}
D_\alpha&(P_1W\|P_2W) \leq G_\alpha(P_1W,P_2W; Q^*W)
\\
&= \frac{\alpha}{1-\alpha}D(Q^*W\|P_1W) + D(Q^*W\|P_2W)
\\
&\leq \frac{\alpha}{1-\alpha}D(Q^*\|P_1) + D(Q^*\|P_2)  = D_\alpha(P_1\|P_2)\,.
\end{align*}
The data processing inequality for the Kullback-Leibler divergence \cite{csiszar_korner} was used in the last inequality.

\item The unique optimizing distribution for $G_\alpha(P_1,P_2;Q)$ is
\begin{equation*}
Q^*(x) = \frac{P_1(x)^\alpha P_2(x)^{1-\alpha}}{\sum_{x\in\m{X}}P_1(x)^\alpha P_2(x)^{1-\alpha}}\,.
\end{equation*}

\qquad\textit{Proof:} Verify by substitution that $G_\alpha(P_1,P_2; Q^*) = D_\alpha(P_1\|P_2)$. Uniqueness follows from strict concavity (resp. convexity) of $G_\alpha(P_1,P_2; Q)$ in $Q$ over $\ms{P}(S(P_1))$ for $\alpha>1$ (resp. $\alpha<1$).

\end{enumerate}

\subsection{$I_\alpha(P,W)$, $K_\alpha(P,W)$ and $C_\alpha(W)$}\label{subsec:C}

\begin{enumerate}[1.]

\item $K_\alpha(P,W)\leq I_\alpha(P,W)$ for $\alpha>1$, and $K_\alpha(P,W)\geq I_\alpha(P,W)$ for $\alpha>1$.

\qquad\textit{Proof:} Immediate from Theorem \ref{thrm} by substituting $Q=P$ in the expressions for $K_\alpha(P,W)$.

\item $I_\alpha(P,W) \leq H(P)$ and $K_\alpha(P,W) \leq  H_\frac{1}{\alpha}(P)$, with equality if and only if $I(P,W)=H(P)$.

\qquad\textit{Proof:} From (\ref{eq:I_a}) we have that for $\alpha>1$
\begin{equation*}
I_\alpha(P,W) \leq \max_V I(P,V) = H(P)
\end{equation*}
A necessary and sufficient condition for an equality is $I(P,V)=H(P)$ and $D(V\|W|P) = 0$ for some $V$, implying $P\circ W = P\circ V$, hence the first assertion. Using this inequality in (\ref{eq:K_a}), along with the $\max$ counterpart of (\ref{eq:H_a}), yields
\begin{align*}
K_\alpha(P,W) \leq  \max_{P'}\left\{H(P')+\frac{1}{1-\alpha}D(P'\|P)\right\}\!= \!H_\frac{1}{\alpha}(P)
\end{align*}
which inherits the same equality condition, hence the second assertion. For $\alpha<1$, substituting $V=W$ in the $\min$ counterpart of (\ref{eq:I_a}) yields
\begin{equation*}
I_\alpha(P,W) \leq I(P,W) \leq H(P)
\end{equation*}
If $I(P,W)=H(P)$ then $I(P,V)=H(P)$ for the minimizing $V$, hence $V(\cdot|x)=W(\cdot|x)$ for $x\in S(P)$ is optimal. The other direction is trivial, hence the first assertion. The second assertion follows similarly as above.

\item $I_\alpha(P,W)$ is concave in $P$ for any fixed $W$ and any $\alpha$, and is convex in $W$ for any fixed $P$ and $\alpha<1$.

\qquad\textit{Proof:} In this case working directly with (\ref{eq:I_a_def}) is much easier. Concavity in $P$ follows as a pointwise minimum of concave (in fact linear) functions in $P$. Convexity in $W$ for $\alpha<1$ follows (using Property \ref{subsec:D}.\ref{prop:D_a_convex}) as a minimization of a convex function in $(Q,W)$ over a convex set.

\item For $\alpha>1$, $K_\alpha(P,W)$ is concave in $P$ for any fixed $W$, and convex in $W$ for any fixed $P$.

\qquad\textit{Proof:} Using (\ref{eq:K_a}) and the previous property, concavity in $P$ follows as a maximization of concave functions in $(P,Q)$ over a convex set. Convexity in $W$ follows as a pointwise maximum of convex functions in $W$.

\item $C_\alpha(W) = \max_P K_\alpha(P,W)$.

\qquad\textit{Proof:} For $\alpha>1$ this is immediate from (\ref{eq:K_a}). The case of $\alpha<1$ does not follow simply from our representation, see \cite{Csiszar95}.

\item (\textit{Data Processing Inequality}) For any distribution $P\in\ms{P}(\m{X})$ and channels $W_1:\m{X}\mapsto\m{Y}, W_2:\m{Y}\mapsto\m{Z}$,
    \begin{align*}
    I_\alpha(P,W_1W_2) &\leq I_\alpha(P,W_1)\\
    K_\alpha(P,W_1W_2) &\leq K_\alpha(P,W_1)
    \end{align*}
    where $W_1W_2$ is the concatenation of the channels $W_1$ and $W_2$, i.e., $(W_1W_2)(z|x) \dfn \sum_yW_2(z|y)W_1(y|x)$.

\qquad\textit{Proof:} Similar to that of Property
\ref{subsec:D}.\ref{prop:DP_div}.

\end{enumerate}

\section{A Composite Hypothesis Testing Problem}\label{sec:hyp}
Suppose two sensors monitor the occurrence of some phenomena. The sensors may generally have different sampling rates with some ratio $\lambda>0$, i.e., for each sample provided by Sensor $1$, $\lambda$ samples are provided by Sensor $2$. When the phenomena is present, it is observed at Sensor $1$ as i.i.d. samples from an unknown distribution $P_1$ in some given family $\mb{P_1}\subseteq \ms{P}(\m{X})$, and at Sensor $2$ as i.i.d. samples from an unknown distribution $P_2$ in some given family $\mb{P_2}\subseteq \ms{P}(\m{X})$. When the phenomena is absent, both sensors observe i.i.d. samples from a common unknown ``ambient noise'' distribution $Q$ in some given family $\mb{Q}\subseteq\ms{P}(\m{X})$. The samples obtained form the sensors are assumed to be mutually independent under each hypothesis.

Suppose we are given $n$ samples from the two Sensors together, where the first $n_1$ samples are from Sensor~$1$, and the last $n_2 = \lambda n_1$ samples\footnote{For brevity, we disregard integer issues.}  are from Sensor~$2$. A \textit{decision rule} corresponds to a set $\Omega_n\subseteq\m{X}^n$, which is allowed to be a function of the families $\mb{P_1},\mb{P_2},\mb{Q}$, but not of the actual $(P_1,P_2,Q)$. The decision rule declares ``phenomena'' if the sample vector lies in $\Omega_n$, and ``no phenomena'' otherwise. The \textit{miss-detection} and \textit{false-alarm} error probabilities associated with $\Omega_n$ and a triplet $(P_1,P_2,Q)$ are
\begin{align*}
p_{\rm \scriptscriptstyle MD}(\Omega_n|P_1,P_2) &\dfn P^{(n)}\left(\m{X}^n\setminus \Omega_n\right)
\\
p_{\scriptscriptstyle FA}(\Omega_n|Q) &\dfn Q^n(\Omega_n)
\end{align*}
where $P^{(n)}\dfn P_1^{n_1}\times P_2^{n_2}$. The \textit{miss-detection exponent} associated with a sequence $\Omega=\{\Omega_n\}_{n=1}^\infty$ of decision rules is
\begin{align*}
E_{\scriptscriptstyle MD}(\Omega|P_1,P_2) &\dfn \liminf_{n\rightarrow \infty} -\frac{1}{n}\log p_{\scriptscriptstyle MD}(\Omega_n|P_1,P_2)\,.
\end{align*}
We will be interested here in maximizing the worst-case mistedection exponent while guaranteing a vanishing false-alarm probability, over all feasible $(P_1,P_2,Q)$. Namely, we will consider
\begin{align*}
E_{\scriptscriptstyle MD}^* &\dfn \sup_{\Omega\in\ms{F}}\inf_{P_1\in\mb{P_1},P_2\in\mb{P_2}}E_{\scriptscriptstyle MD}(\Omega|P_1,P_2)
\end{align*}
where
\begin{equation*}
\ms{F} \dfn \left\{\Omega : \lim_{n\rightarrow\infty}p_{\scriptscriptstyle FA}(\Omega_n|Q) = 0 \,,\;\forall Q\in\mb{Q}\right\}\,.
\end{equation*}

In what follows, let $\delta_n \dfn \frac{|\m{X}|\log{n}}{n}$, and for any two families $\mb{P},\mb{P'}\subseteq\ms{P}(\m{X})$, define
\begin{equation}
D_\alpha(\mb{P}\|\mb{P'}) \dfn \inf_{P\in\mb{P},P'\in\mb{P'}}D_\alpha(P\|P')\,.
\end{equation}
Furthermore, write $\mb{Q}^*$ for the closure of the family of all distributions of the form
\begin{equation*}
Q^*(x) = \frac{P_1(x)^\frac{1}{1+\lambda} P_2(x)^{\frac{\lambda}{1+\lambda}}}{\sum_{x\in\m{X}}P_1(x)^\frac{1}{1+\lambda} P_2(x)^{\frac{\lambda}{1+\lambda}}}
\end{equation*}
for some $P_1\in\mb{P_1}, P_2\in\mb{P_2}$.

\begin{example}\label{ex:single}
The case where $\lambda=0$ (single sensor) corresponds to a classical setting of composite hypothesis testing. It is well known that in this case \cite{CsiszarShields}
\begin{equation*}
E_{\scriptscriptstyle MD}^*=D(\mb{Q}\|\mb{P_1})
\end{equation*}
which can be achieved by the decision rule
\begin{equation}\label{eq:simple_rule}
\Omega_n = \left\{x^n : \inf_{Q\in\mb{Q}}D(\pi_{x^n}\|Q) \geq \delta_n\right\}\,.
\end{equation}
\end{example}

\begin{example}
If $\mb{P_1}\cap \mb{P_2}\cap \mb{Q}\neq \emptyset$, then $E_{\scriptscriptstyle MD}^* = 0$ for any $\lambda$.
\end{example}
\begin{example}
Suppose $\mb{P_1}$ and $\mb{P_2}$ have disjoint supports, i.e., $S(P_1)\cap S(P_2) = \emptyset$ for all $P_1\in\mb{P_1}$ and $P_2\in\mb{P_2}$. Then $E_{\scriptscriptstyle MD}^* = \infty$ regardless of $\mb{Q}$. This is achieved by a simple decision rule that declares ``phenomena'' when the empirical supports of the samples from the sensors are disjoint, and ``no phenomena'' otherwise. Clearly, this rule has a zero miss-detection probability for any $n$. It is also easy to see that its false-alarm probability tends to zero exponentially for any $Q\in\ms{P}(\m{X})$.
\end{example}

Generally, one would expect the optimal miss-detection exponent to be related to some measure of disparity between the families $\mb{P_1}$ and $\mb{P_2}$, quantifying the fact that the noise $Q$ cannot mimic both $P_1$ and $P_2$ too well at the same time. As it turns out, at least in the worst case sense over the choice of $\mb{Q}$, this measure is related to a R\'enyi divergence between the two families.
\begin{theorem}\label{thrm2}
For any choice of $\mb{P_1,P_2,Q}$ and $\lambda$,
\begin{equation*}
E_{\scriptscriptstyle MD}^*  \geq  \lambda(1+\lambda)^{-1}D_{\frac{1}{1+\lambda}}(\mb{P_1}\|\mb{P_2})
\end{equation*}
with equality if and only if the closure of $\mb{Q}$ has an nonempty intersection with the associated $\mb{Q}^*$.
\end{theorem}
\begin{proof}
Consider first the case where $\mb{Q}=\{Q\}$. Let us show that
\begin{equation*}
E_{\scriptscriptstyle MD}^*=(1+\lambda)^{-1}\left(D(Q\|\mb{P_1}) + \lambda D(Q\|\mb{P_2})\right)\,.
\end{equation*}
Achievability follows by letting $\Omega_{n_1}^{(1)}$ and $\Omega_{n_2}^{(2)}$ be the optimal per-sensor decision rules as in (\ref{eq:simple_rule}), and setting
\begin{equation}\label{eq:dec_rule1}
\Omega_n \dfn \left\{(x^{n_1},y^{n_2}) : x^{n_1}\in \Omega_{n_1}^{(1)} \;\;\text{or}\;\;y^{n_2}\in \Omega_{n_2}^{(2)}\right\}\,.
\end{equation}
The converse is a simple generalization of the standard single-sensor case \cite{CsiszarShields}. Let $\Omega'=\{\Omega_n'\}$ be any sequence of decision rules achieving a vanishing false-alarm probability. For $i\in\{1,2\}$, let $\Gamma_{n_i}$ denote the union of all $n_i$-dimensional type classes $T_{Q_i}$ where $Q_i\in\ms{P}^{n_i}(\m{X})$ satisfies $D(Q_i\|Q) \leq \delta_{n_i}$. By Lemma \ref{lem:types} property (\ref{item:large_dev}), we have $Q^n(\Gamma_{n_1}\times \Gamma_{n_2})\rightarrow 1$ as $n\rightarrow\infty$. Since by our assumption $Q^n(\m{X}^n\setminus \Omega_n')\rightarrow 1$, then $Q^n((\Gamma_{n_1}\times \Gamma_{n_2})\setminus \Omega_n')\geq \frac{1}{2}$ (say) for any $n$ large enough. Thus, there must exist a pair of types $(Q_{1,n},Q_{2,n})\in \Gamma_{n_1}\times \Gamma_{n_2}$ such that $Q^n((T_{Q_{1,n}}\times T_{Q_{2,n}})\setminus \Omega_n')\geq \frac{1}{2}Q^n(T_{Q_{1,n}}\times T_{Q_{2,n}})$. Since both $Q^n$ and $P^{(n)}$ are constant over $T_{Q_{1,n}}\times T_{Q_{2,n}}$, the same inequality holds for $P^{(n)}$. Therefore,
\begin{align*}
\begin{multlined}
-\frac{1}{n}\log P^{(n)}(\m{X}^n\setminus\Omega_n')
\\
\shoveleft[30pt]{\leq -\frac{1}{n}\log P^{(n)}((T_{Q_{1,n}}\times T_{Q_{2,n}})\setminus \Omega_n')}
\\
\shoveleft[30pt]{\leq -\frac{1}{n}\log \frac{1}{2}P^{(n)}(T_{Q_{1,n}}\times T_{Q_{2,n}})}
\\
\shoveleft[30pt]{\leq  (1+\lambda)^{-1}\left(D(Q_{1,n}\|P_1) + \lambda D(Q_{2,n}\|P_2)\right)}
\\
\shoveright{+\frac{1+2|\m{X}|\log{(n+1)}}{n}}\\ \vspace{-15pt}
\end{multlined}
\end{align*}
where properties (\ref{item:prob_in_type})-(\ref{item:num_of_types}) of Lemma \ref{lem:types} were used in the last inequality. Letting $n\rightarrow\infty$, and recalling that $D(Q_{i,n}\|Q)\rightarrow 0$ which implies $D(Q_{i,n}\|P_i)\rightarrow D(Q\|P_i)$, the converse follows.

As a result, it is now clear that for a general $\mb{Q}$
\begin{equation}\label{eq:E_md1}
E_{\scriptscriptstyle MD}^* \leq (1+\lambda)^{-1}\inf_{Q\in\mb{Q}}\left(D(Q\|\mb{P_1}) + \lambda D(Q\|\mb{P_2})\right)\,.
\end{equation}
The decision rule (\ref{eq:dec_rule1}) above (with $\Omega_{n_1}^{(1)}$ and $\Omega_{n_2}^{(2)}$ now taking the infimum over the family $\mb{Q}\,$) will generally fail to achieve the upper bound in (\ref{eq:E_md1}), and may even not attain a vanishing miss-detection probability. For instance, if $\mb{P_1}=\{P_1\}$, $\mb{P_2}=\{P_2\}$ and $\mb{Q}=\{P_1,P_2\}$, then $p_{\rm \scriptscriptstyle MD}(\Omega_n|P_1,P_2) \rightarrow 1$, whereas the upper bound (\ref{eq:E_md1}) is positive if $P_1\neq P_2$. Clearly, the problem is that each sensor makes its own binary decision before those are combined, not taking into account that $Q$ is common. This shortcoming is easily corrected by the following modified decision rule:
\begin{align*}
\wt{\Omega}_n \!= &\left\{(\!x^{n_1},y^{n_2}) : \!\inf_{Q\in\mathbf{Q}}\!\max \left\{D(\pi_{x^{n_1}}\|Q),D(\pi_{y^{n_2}}\|Q)\right\} \geq \delta'_n\!\right\}\
\end{align*}
where $\delta'_n = \max(\delta_{n_1},\delta_{n_2})$.

Let us show that this rule attains the upper bound in (\ref{eq:E_md1}). For any $Q\in\mb{Q}$, $\wt{\Omega}_n$ is contained in the set of all vectors $(x^{n_1},y^{n_2})$ for which either $D(\pi_{x^{n_1}}\|Q) \geq \delta'_n$ or $D(\pi_{y^{n_2}}\|Q) \geq \delta'_n$. Thus, using Lemma \ref{lem:types} property (\ref{item:large_dev}) together with the union bound, we obtain
\begin{align*}
p_{\scriptscriptstyle FA}(\wt{\Omega}_n|Q) &\leq |\ms{P}^{n_1}(\m{X})|2^{-n_1\delta'_n} + |\ms{P}^{n_2}(\m{X})|2^{-n_2\delta'_n}
\\
&\leq \!{n_1+|\m{X}|-1 \choose |\m{X}|-1}n_1^{-|\m{X}|} \!+\! {n_2+|\m{X}|-1 \choose |\m{X}|-1}n_2^{-|\m{X}|}
\end{align*}
hence $p_{\scriptscriptstyle FA}(\wt{\Omega}_n|Q)\rightarrow 0$ as $n\rightarrow\infty$, for any $Q\in\mb{Q}$.

Define the set $\Pi_n\subseteq\ms{P}^{n_1}(\m{X})\times \ms{P}^{n_2}(\m{X})$ of all the type pairs $(Q_1,Q_2)$ for which there exists some $Q\in\mb{Q}$ such that $D(Q_1\|Q) < \delta'_n$ and $D(Q_2\|Q) < \delta'_n$. By definition, $\m{X}^n\setminus \wt{\Omega}_n$ is a union of all type classes products pertaining to $\Pi_n$. Therefore, using properties (\ref{item:prob_in_type})-(\ref{item:large_dev}) of Lemma \ref{lem:types} again, we get
\begin{align*}
\begin{multlined}
-\frac{1}{n}\log P^{(n)}(\m{X}^n\setminus\wt{\Omega}_n)
\\
\shoveleft[30pt]{=  -\frac{1}{n}\log\sum_{(Q_1,Q_2)\in\Pi_n}P_1^{n_1}\left(T_{Q_1}\right)\cdot P_2^{n_2}\left(T_{Q_2}\right)}
\\
\shoveleft[30pt]{\geq  (1+\lambda)^{-1}\hspace{-14pt}\min_{(Q_1,Q_2)\in\Pi_n}\hspace{-5pt}\left(D(Q_1\|P_1) + \lambda D(Q_2\|P_2)\right)}\\
\shoveright{-\frac{2|\m{X}|\log{(n+1)}}{n}\,.}\\ \vspace{-15pt}
\end{multlined}
\end{align*}
Let $(Q_{1,n},Q_{2,n})$ achieve the minimum above. Then by definition there exists $Q_n\in\mb{Q}$ such that $D(Q_{i,n}\|Q_n)<\delta'_n\rightarrow 0$ for $i\in\{1,2\}$, which implies that $D(Q_{n,i}\|P_i) \rightarrow D(Q_n\|P_i)$. Hence for any $P_1\in\mb{P_1},P_2\in\mb{P_2}$,
\begin{equation*}
E_{\scriptscriptstyle MD}(\wt{\Omega}|P_1,P_2) \geq (1+\lambda)^{-1}\inf_{Q\in\mb{Q}}\left(D(Q\|P_1) + \lambda D(Q\|P_2)\right)
\end{equation*}
Therefore, $\wt{\Omega}$ attains the upper bound in (\ref{eq:E_md1}), and thus
\begin{align}\label{eq:E_md2}
E_{\scriptscriptstyle MD}^* &= (1+\lambda)^{-1}\inf_{Q\in\mb{Q}}\left(D(Q\|\mb{P_1}) + \lambda D(Q\|\mb{P_2})\right)
\\
\nonumber &\geq (1+\lambda)^{-1}\min_{Q\in\ms{P}(\m{X})}\left(D(Q\|\mb{P_1}) + \lambda D(Q\|\mb{P_2})\right)
\\
\nonumber & = \lambda(1+\lambda)^{-1}D_{\frac{1}{1+\lambda}}(\mb{P_1}\|\mb{P_2})
\end{align}
where the inequality is on account of Theorem \ref{thrm}.\footnote{Note that for the $\alpha<1$ counterpart of (\ref{eq:D_a}), minimizing over $Q\in\ms{P}(\m{X})$ instead of $Q\ll P_1$ changes nothing.} Property \ref{subsec:D}.\ref{prop:opt_D} verifies the necessary and sufficient conditions for an equality.
\end{proof}

The lower bound in Theorem \ref{thrm2} is independent of the noise family $\mb{Q}$, hence the R\'enyi divergence between the families $\mb{P_1}$ and $\mb{P_2}$ admits an operational interpretation as the optimal worst-case miss-detection exponent (up to a constant) when the noise distribution $Q$ is completely unknown (i.e., $\mb{Q}=\ms{P}(\m{X})$), or more generally, when $Q$ can take values in the ``worst noise'' set $\mb{Q}^*$. In other cases this serves only as a lower bound, and the strictly larger exponent is given by (\ref{eq:E_md2}). It is possible (somewhat artificially) to interpret this exponent as a (limit of a) generalized form of the R\'enyi divergence, taking into account also the family $\mb{Q}$, as we now proceed to show.

Let $(\alpha_1,\ldots,\alpha_{k+1})$ be a probability vector, and write $\underline{\alpha}\dfn (\alpha_1,\ldots,\alpha_k)$. Let $\{P_1,\ldots,P_{k+1}\}$ be distributions over $\ms{P}(\m{X})$. We define the associated \textit{generalized R\'enyi divergence of order $\underline{\alpha}$} to be
\begin{equation*}
D_{\underline{\alpha}}(P_1,\ldots,P_{k+1}) \dfn -\log \left(\sum_{x\in\m{X}}\prod_{i=1}^{k+1}P_i(x)^{\alpha_i}\right)\,.
\end{equation*}
For families of distributions $\{\mb{P_1},\ldots,\mb{P_{k+1}}\}$, we define
\begin{equation*}
D_{\underline{\alpha}}(\mb{P_1},\ldots,\mb{P_{k+1}}) \dfn \inf_{\{P_i\in\mb{P_i}\}}D_{\underline{\alpha}}(P_1,\ldots,P_{k+1})\,.
\end{equation*}
Additivity of the generalized R\'enyi divergence is easily verified, which leads to
\begin{corollary}
\begin{equation*}
D_{\underline{\alpha}}(P_1,\ldots,P_{k+1})  = \min_{Q\in\ms{P}(\m{X})}\sum_{j=1}^{k+1}\alpha_jD(Q\|P_j)\,.
\end{equation*}
\end{corollary}

\begin{theorem}
For any $0<\alpha\leq (1+\lambda)^{-1}$,
\begin{equation*}
E_{\scriptscriptstyle MD}^*  \geq  (1+\lambda)^{-1}\alpha^{-1}D_{\left(\alpha,\lambda\alpha\right)}(\mb{P_1},\mb{P_2},\mb{Q}) \dfn E_{\scriptscriptstyle MD}^*(\alpha)
\end{equation*}
Furthermore, $E_{\scriptscriptstyle MD}^*(\alpha)$ is monotonically non-increasing in $\alpha$, and if $E_{\scriptscriptstyle MD}^*<\infty$ then
\begin{equation*}
E_{\scriptscriptstyle MD}^*  = \lim_{\alpha\rightarrow 0^+}E_{\scriptscriptstyle MD}^*(\alpha)
\end{equation*}
\end{theorem}
\begin{proof}
\begin{align*}
\begin{multlined}
E_{\scriptscriptstyle MD}^* = (1+\lambda)^{-1}\inf_{Q\in\mb{Q}}\left(D(Q\|\mb{P_1}) + \lambda D(Q\|\mb{P_2})\right)
\\
\shoveleft[\widthof{$E_{\scriptscriptstyle MD}^*$}]{\geq (1+\lambda)^{-1}\min_{Q'\in\ms{P}(\m{X})}\left(D(Q'\|\mb{P_1}) + \lambda D(Q'\|\mb{P_2})\right.}
\\
\shoveright{\left.+ \left(\alpha^{-1}-(\lambda+1)\right)D(Q'\|\mb{Q})\right)}
\\
\shoveleft[\widthof{$E_{\scriptscriptstyle MD}^*$}]{ = (1+\lambda)^{-1}\alpha^{-1}D_{\left(\alpha,\lambda\alpha\right)}(\mb{P_1},\mb{P_2},\mb{Q})\,.}\\\vspace{-10pt}
\end{multlined}
\end{align*}
Monotonicity is clear from the second line above. Tightness in the limit is proved in a similar way to Property \ref{subsec:H}.\ref{prop:H_1}, by noting that $E_{\scriptscriptstyle MD}^*<\infty$ implies $D(Q'\|\mb{Q})\rightarrow 0$ as $\alpha\rightarrow 0$ for the optimizing $Q'$.
\end{proof}

\end{document}